\documentclass[runningheads,12pt]{llncs}

\usepackage{amsmath,amsfonts,amssymb}

\usepackage{subfigure}
\usepackage{algorithm,caption}
\usepackage{algpseudocode}
\usepackage{graphicx,color}
\usepackage{verbatim}
\usepackage[bookmarks=false]{hyperref}

\newcommand{\Gr}{Gr\"{u}nbaum }
\newcommand{\keywords}[1]{\par\addvspace\baselineskip
\noindent\keywordname\enspace\ignorespaces#1}

\begin{document}

\title{A New Rose : The First Simple Symmetric 11-Venn Diagram}
\author{Khalegh Mamakani
\and Frank Ruskey}
\institute{Dept. of Computer Science, University of Victoria,
Canada.}
\authorrunning{Mamakani and Ruskey}

\maketitle

\begin{abstract}
A symmetric Venn diagram is one that is invariant under rotation, up to a relabeling of curves.
A simple Venn diagram is one in which at most two curves intersect at any point.
In this paper we introduce a new property of Venn diagrams called crosscut symmetry, which is
  related to dihedral symmetry.
Utilizing a computer search restricted to crosscut symmetry we found many simple symmetric Venn diagrams
  with 11 curves.
This answers an existence question that has been open since the 1960's.
The first such diagram that was discovered is shown here.

\end{abstract}

\section{Introduction}

Mathematically, an \emph{$n$-Venn diagram} is a collection of $n$ simple closed curves in the plane with the
following properties: (a) Each of the $2^n$ different intersections of the open interiors or exteriors of the curves is
a non-empty connected region; (b) there are only finitely many points where the curves intersect.
If each of the intersections is of only two curves, then the diagram is said to be \emph{simple}.
A \emph{$k$-region} ($0 \le k \le n$) in an $n$-Venn diagram is a region which is in the interior of exactly $k$ curves.
In a \emph{monotone} Venn diagram every $k$-region is adjacent to at least one $(k-1)$-region (if $k>0$) and
  it is also adjacent to at least one $(k+1)$-region (if $k < n$).
Monotone Venn diagrams are precisely those that can be drawn with convex curves \cite{BGR_1999}.
The diagrams under consideration in this paper are both monotone and simple.

A $n$-Venn diagram is \emph{symmetric} if it is left fixed (up to a relabeling of the curves) by a rotation of the
  plane by $2\pi/n$ radians.
Interest in symmetric Venn diagrams was initiated by Henderson in a 1963 paper in which he showed that a symmetric
  $n$-Venn diagram could not exist unless $n$ is a prime number \cite{He_1963} (see also \cite{WW_2008}).
Of course, it is easy to draw symmetric $2$- and $3$-Venn diagrams using circles as the curves, but it was not
  until 1975 that Gr\"{u}nbaum \cite{Gr_1975} published a simple symmetric $5$-Venn diagram, one that could be drawn using ellipses.
Some 20 years later, in 1992, simple symmetric $7$-Venn diagrams were discovered independently by Gr\"{u}nbaum \cite{Gr_2_1992} and by
  Edwards \cite{Ed_1998}.
The total number of non-isomorphic simple $7$-Venn diagrams that are convexly drawable is also known \cite{MMR_2012}.

However, the construction of a simple symmetric 11-Venn diagram has eluded all previous efforts until now.
We know of several futile efforts that involved either incorrect constructions, or unsuccessful computer searches.

It should be noted that if the diagrams are not constrained to be simple, then Hamburger \cite{Ha_2002} was the first to
discover a (non-simple) symmetric 11-Venn diagram, and Griggs, Killian and Savage (GKS) have shown how to construct symmetric, but highly
  non-simple, $n$-Venn diagrams whenever $n$ is prime \cite{GKS_2004}.
These constructions, in a sense, are maximally non-simple since they involve points where all $n$ curves intersect.
Some progress towards ``simplifying" the GKS construction is reported in \cite{KRSW_2004}, but their approach could never succeed in
  producing truly simple diagrams.

Part of the interest in Venn diagrams is due to the fact that their geometric dual graphs are planar spanning subgraphs
  of the hypercube; furthermore, if the Venn diagram is simple then the subgraph is maximum in the sense that every
  face is a quadrilateral.
Symmetric drawings of Venn diagrams imply symmetric drawings of spanning subgraphs of the hypercube.
Some recent work on finding symmetric structures embedded in the hypercube is reported in
  \cite{Jordan_2010} and \cite{DMT_2012}.

\section{Crosscut symmetry}

 We define a \emph{crosscut} of a Venn diagram as a segment of a curve which sequentially ``cuts"
   (i.e., intersects) every other curve without repetition.
 Except for $n=2$ and $n=3$, where the symmetric 2-Venn and 3-Venn diagrams have 4 and 6 crosscuts respectively,
   a symmetric $n$-Venn diagram either has $n$ crosscuts or it has none.
 Referring to Figure \ref{Fig:VD_7:a}, notice that each of the 7 curves has a crosscut.

 \begin{lemma}
 If $n > 3$, then a symmetric $n$-Venn diagram has at most one crosscut per curve.
 \end{lemma}

 \begin{proof}
 A curve of a Venn diagram touches a face at most once.
 Thus a curve in any Venn diagram can not have three or more crosscuts, because that curve
 would touch the outer face (and the innermost face) at least twice.  Now suppose that some
 curve $C$ in an $n$-Venn diagram has two crosscuts.  Then those crosscuts must start at the same segment of $C$
 on the outer face, and finish at the same segment of $C$ on the innermost face.
 Thus, curve $C$ contains exactly $2(n-1)$ intersections with the other curves.
 If the Venn diagram is symmetric, then there must be a total of $n(n-1)$ intersection points.
 On the other hand, a simple symmetric Venn diagram has exactly $2^n-2$ intersection points.
 Since $n(n-1) = 2^n-2$ has a solution for $n = 1,2,3$, but not for $n > 3$, the lemma is proved.
 \qed
 \end{proof}

\begin{figure}
\centering{
\subfigure[]{
\includegraphics[scale=.25]{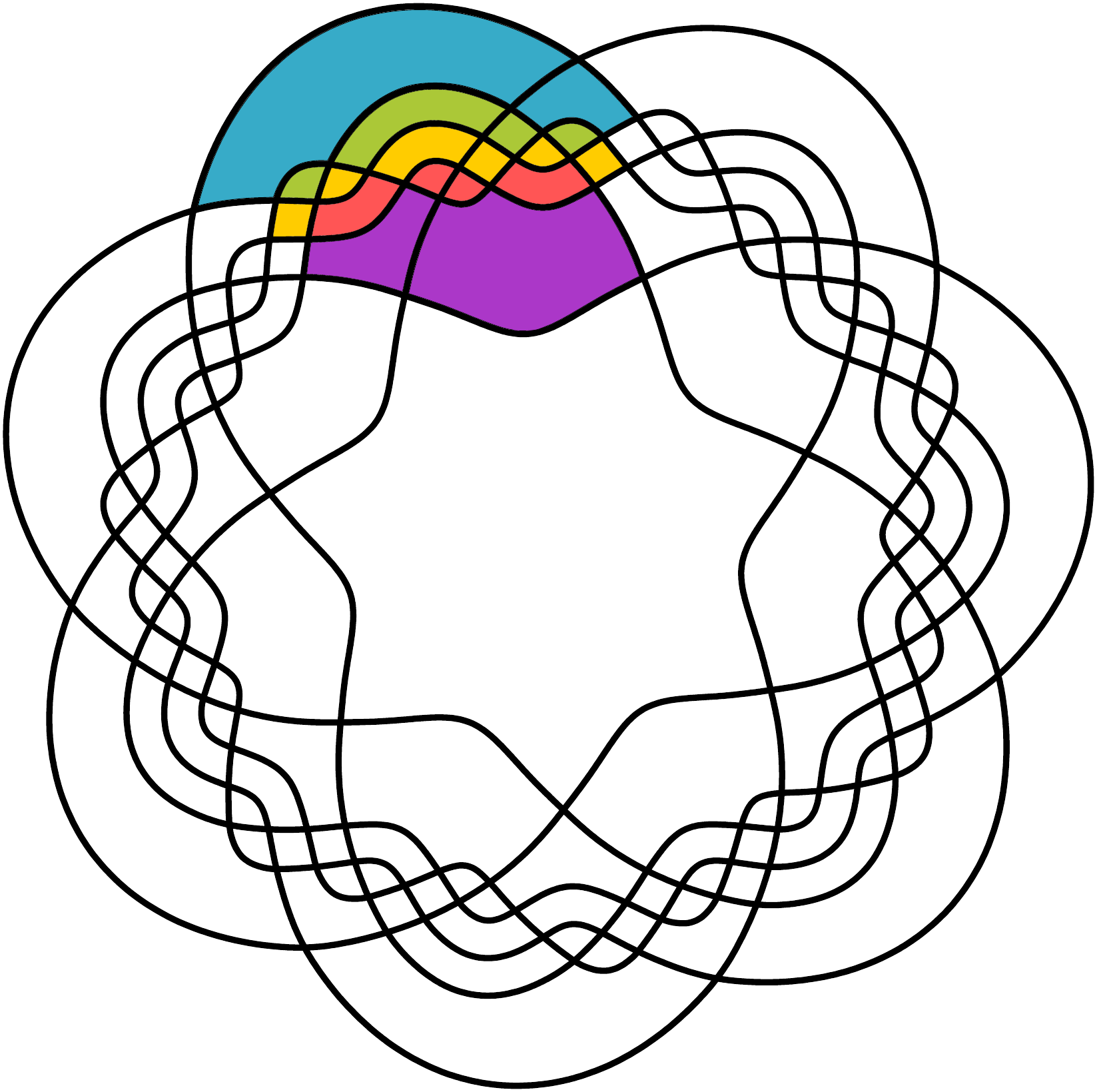}
\label{Fig:VD_7:a}
}
\subfigure[]
{
\includegraphics[scale=.70]{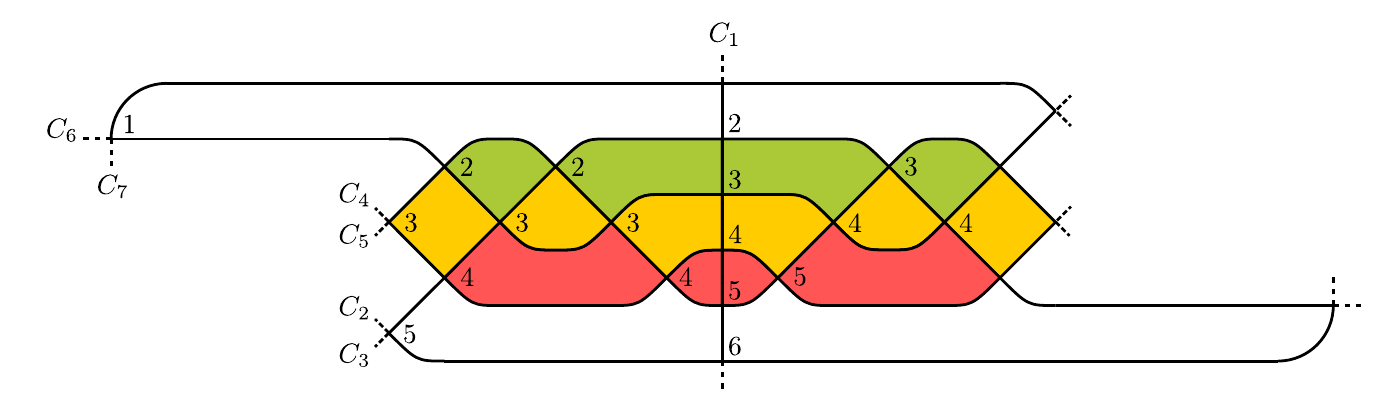}
\label{Fig:VD_7:b}}
}
\caption{\subref{Fig:VD_7:a} A simple rotationally symmetric monotone 7-Venn diagram with a cluster colored (shaded).
\subref{Fig:VD_7:b} The cylindrical representation of a cluster of the diagram, showing the reflective
aspect of crosscut symmetry.}
\label{Fig:VD_7}
\end{figure}

A \emph{clump} in a Venn diagram is a collection of regions that is bounded by a simple closed path of curve segments.
The \emph{size} of a clump is the number of regions that it contains.
Aside from the innermost face and the outermost face, a rotationally symmetric $n$-Venn diagram can be partitioned
   into $n$ congruent clumps, each of size $(2^n-2)/n$; in this case we call the
   clump a \emph{cluster} --- it is like a fundamental region for the rotation, but omitting the parts of the fundamental region
   corresponding to the full set and to the empty set.
Referring again to Figure \ref{Fig:VD_7:a}, a cluster has been shaded, and this cluster is redrawn in Figure \ref{Fig:VD_7:b}.
Notice that the cluster of Figure \ref{Fig:VD_7:b} has a central shaded section which has a reflective symmetry about the crosscut.
The essential aspects of this reflective symmetry are embodied in the definition of crosscut symmetry, given below.

\begin{definition}
 Given a rotationally symmetric $n$-Venn diagram, we label the curves as $C_1, C_2, \ldots, C_n$ according to the clockwise order in which
   they touch the unbounded outermost face.
 Assume that we have a cluster with the property that every curve intersects the cluster in a segment and that
   cluster $S_k$ contains the crosscut for curve $C_k$.
 Let $L_{i,k}$ be the list of curves that we encounter as we follow $C_i$ in the cluster $S_k$ in clockwise order, and
   let $\ell_{i,k}$ denote the length of $L_{i,k}$.
 A rotationally symmetric $n$-Venn diagram has \emph{crosscut symmetry} if it can be partitioned into $n$ such clusters
   $S_1, S_2, \ldots , S_n$ in such a way that for every cluster $S_k$, for any $i \ne k$,
   the list $L_{i,k}$ is palindromic, that is, for $1 \le j \le \ell_{i,k}$, we have $L_{i,k}[j] = L_{i,k}[\ell_{i,k}-j+1]$.
\end{definition}

Figure \ref{Fig:VD_7:a} shows a simple rotationally symmetric 7-Venn diagram which also has crosscut symmetry.
In the Survey of Venn diagrams it is known as M4 \cite{RW_2004}.
There is only one other simple 7-Venn diagram with crosscut symmetry; it is dubbed ``Hamilton" by Edwards \cite{Ed_1998}.
Hamilton also has ``polar symmetry" but is not used here because it might cause confusion with the crosscut symmetry.
Cluster $S_1$ of the diagram is shown in Figure \ref{Fig:VD_7:b} where a segment of $C_1$ is the crosscut.
The list of crossing curves for each curve in the cluster $S_1$ is shown below.
\begin{alignat}{1}
L_{1,1} &= [C_2,C_5, C_4, C_6, C_3, C_7]\notag\\
L_{2,1} &= [C_3,C_1,C_3]\notag\\
L_{3,1} &= [C_2,C_4,C_6,C_5,C_1,C_5,C_6,C_4,C_2]\notag\\
L_{4,1} &= [C_5,C_3,C_5,C_1,C_5,C_3,C_5]\notag\\
L_{5,1} &= [C_4,C_6,C_3,C_6,C_4,C_1,C_4,C_6,C_3,C_6,C_4]\notag\\
L_{6,1} &= [C_7,C_5,C_3,C_5,C_1,C_5,C_3,C_5,C_7]\notag\\
L_{7,1} &= [C_6,C_1,C_6]\notag
\end{alignat}

Imagine a ray issuing at some point in the innermost region of an $n$-Venn diagram $V$ and sweeping the surface of the diagram in a clockwise order. If $V$ is simple and monotone, we can always deform the curves of $V$ continuously such that at any moment the ray cuts each curve exactly once and no two intersection points of $V$ lie on the ray at the same time \cite{BGR_1999}. Let $\pi$ be the vector of curve labels along the ray where $\pi(1)$ is the outermost curve and $\pi(n)$ is the innermost curve.
At each intersection point of $V$ a pair of adjacent curves cross. Therefore, we can represent the whole diagram using a sequence of length $2^n-2$ of curve crossings, where a crossing of curves $\pi(i)$ and $\pi(i+1)$, for $1 \le i < n$, is indicated by an entry of value $i$ in the \emph{crossing sequence}. See \cite{MMR_2012} for more details on this representation of simple monotone Venn diagrams. For a simple rotationally symmetric $n$-Venn diagram, the first $(2^n-2)/n$ elements of the crossing sequence are enough to represent the entire diagram; the remainder of the crossing sequence is formed by $n-1$
concatenations of this sequence.
For example, if the "ray" initially starts a small distance from the left boundary of the cluster shown in Figure \ref{Fig:VD_7:a}, then the
first $18 = (2^7-2)/7$ elements of the crossing sequence could be the following list; we write ``could" here since crossings can sometimes occur
in different orders and still represent the same diagram (e.g., taking $\rho =3,1,5,2,4$ gives the same diagram); see the remark below.
\[
\underbrace{1,3,2,5,4}_\rho,\underbrace{3,2,3,4}_\alpha,\underbrace{6,5,4,3,2}_\delta,\underbrace{5,4,3,4}_{\alpha^{r+}}
\]
We encourage the reader to verify that this sequence is correct by referring to the crossing numbers shown to the right of the
intersections in Figure \ref{Fig:VD_7:b}.

\begin{remark}
\label{rem:swap}
If $j,k$ is an adjacent pair in a crossing sequence $\mathcal{C}$ and $|j-k| > 1$, then
the sequence $\mathcal{C'}$ obtained by replacing the pair $j,k$ with the pair $k,j$ is also
a crossing sequence of the same diagram.
\end{remark}

\begin{theorem}\label{Thm:Crosscut}
A simple monotone rotationally symmetric $n$-Venn diagram is crosscut symmetric if and only if it can be represented by a crossing sequence of the form $\rho, \alpha, \delta, \alpha^{r+}$ where
\begin{itemize}
\item $\rho$ is $1,3,2,5,4, \ldots, n-2, n-3$.
\item $\delta$ is $n-1, n-2, \ldots, 3, 2$.
\item $\alpha$ and $\alpha^{r+}$ are two sequences of length $(2^{n-1}-(n-1)^2)/n$ such that $\alpha^{r+}$ is obtained by reversing $\alpha$ and adding $1$ to each element; that is, $\alpha^{r+}[i]=\alpha[|\alpha|-i+1]$.
\end{itemize}
\end{theorem}

\begin{proof}
Let $V$ be a simple monotone rotationally symmetric $n$-Venn diagram with crosscut symmetry.  Assume that the curves of $V$ are labeled $C_1, C_2,\ldots, C_n$ according to their clockwise order of touching the outermost region. Consider a cluster of $V$ where a segment of $C_1$ forms the crosscut. Each curve of the diagram touches the outermost and innermost region exactly once. Therefore, $C_1$ only intersects $C_{n-1}$ at some point $u$ on the left border of the cluster before intersecting the crosscut.
Since the diagram is crosscut symmetric, $L_{n,1} = [C_{n-1},1,C_{n-1}]$ and so $C_n$ must intersect $C_{n-1}$ at some point
  $v$ on the right border of the cluster immediately below $C_1$.
Because of the rotational symmetry of $V$, the point $v$ is the image of some point $s$ on the left border under the rotation of  $2\pi/n$ about the center of the diagram.
At the point $s$ the curves $C_{n-1}$ and $C_{n-2}$ intersect.
Again, because of the crosscut symmetry, $L_{n-1,1} = [C_n,C_{n-2},\ldots,C_{n-2},C_n]$, and so there is a corresponding point $t$ on the right border
  where $C_{n-1}$ and ${C_{n-2}}$ intersect.
Continuing in this way, we can see that the the left and right border of a crosscut symmetric $n$-Venn diagram must have the ``zig-zag" shape as illustrated in Figure \ref{Fig:Crosscut}.

\begin{figure}
\centering{
\includegraphics[scale=0.4]{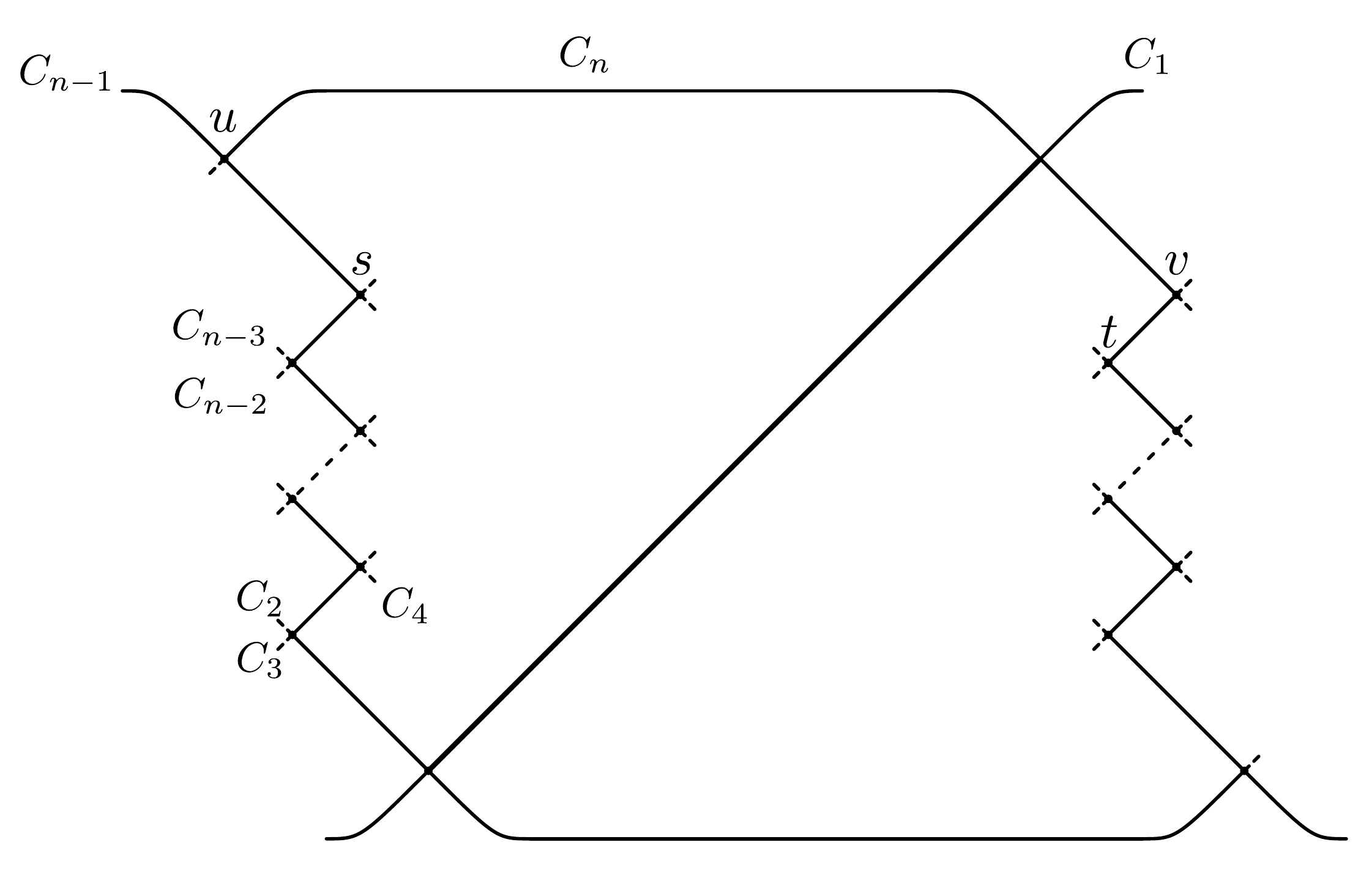}
}
\caption{A cluster of a crosscut symmetric simple monotone $n$-Venn diagram (cylindrical representation).}
\label{Fig:Crosscut}
\end{figure}

Consider two rays issuing from the center of the diagram which cut all the curves immediately before the left zig-zag border and immediately after
  the right zig-zag border of the cluster.
Let $\pi$ and $\sigma$ be the curve vectors along the rays as we move the rays in opposite directions towards the crosscut.
We prove by induction that as long as the two rays do not intersect the crosscut, at each moment $\sigma$ is a cyclic shift of $\pi$ one element to right;
   therefore if $\pi(i)$ and $\pi(i+1)$ intersect on the left side, then $\sigma(i+1)$ and $\sigma(i+2)$ must also intersect on the right side.
Initially $\pi$ and $\sigma$ are
\[
[C_{n-1}, C_{n}, C_{n-3}, C_{n-2}, \ldots, C_2, C_3, C_1] \text{ and } [C_1,C_{n-1}, C_{n}, C_{n-3}, \ldots, C_5, C_2, C_3]
\]
respectively. According to the ``zig-zag" shape of the borders, after $n-2$ crossings on the borders, $\pi$ and $\sigma$ will be
\[[C_{n}, C_{n-2}, C_{n-1}, C_{n-4}, \ldots, C_4, C_2, C_1] \text{ and } [C_1,C_{n}, C_{n-2}, C_{n-1}, \ldots, C_3, C_4, C_2].\]
Clearly, the $n-2$ crossings on the left border can be represented by the crossing sequence $\rho = 1,3,2,5,4, \ldots, n-2, n-3$.

Now suppose for any of the previous $k$ intersection points, $\sigma$ is always a cyclic rotation of $\pi$ one element to the right and suppose the next crossing occurs between $\pi(i)$ and $\pi(i+1)$. Since $\pi(i)=\sigma(i+1)$ and $\pi(i+1)=\sigma(i+2)$ the next crossing on the right side of crosscut must occur between $\sigma(i+1)$ and $\sigma(i+2)$, for otherwise the diagram would not be crosscut symmetric. Thus, $\sigma$ remains a cyclic shift of $\pi$ by one element to the right after the crossing. Therefore, by induction each crossing of $\pi(i)$ and $\pi(i+1)$ corresponds to a crossing of $\sigma(i+1)$ and $\sigma(i+2)$. Let $\alpha$ represent the sequence of crossings that follow the first $n-2$ crossings that occur on the left border.
Reversing the crossings of the right side of the crosscut, the entire crossing sequence of the cluster  is
\[
1,3,2,5,4,\ldots, n-2, n-3, \alpha, n-1, n-2, \ldots, 2,1, \alpha^{r+}, n-2, n-1,\ldots, 5,6,3,4,2,
\]
where $\alpha^{r+}[i]=\alpha[|\alpha|-i+1]$. Removing the $n-1$ elements representing the intersection points of the right border, we will get the required crossing sequence $\rho, \alpha, \delta, \alpha^{r+}$ of the diagram.

To prove the converse, suppose we are given a rotationally symmetric $n$-Venn diagram $V$ with the crossing sequence $\rho, \alpha, \delta, \alpha^{r+}$
  as specified in the statement of the theorem.
Then the crossing sequence of one cluster of $V$ is $\rho,\alpha,\delta,\alpha^{r+}, \rho,n-1$ where $\rho = [1,3,2,5,4, \ldots, n-2, n-3]$,
  and $\delta = [n-1, n-2, \ldots, 3,2]$.  The $\rho$ sequence indicates that the cluster has the ``zig-zag"-shaped borders.
Therefore, by Remark \ref{rem:swap}, the crossing sequence of the cluster can be transformed into an equivalent crossing sequence
\[
A=[1,3,2, \ldots, n{-}2, n{-}3, \alpha, n{-}1, n{-}2, \ldots, 3,2,1,\alpha^{r+}, n-2, n-3, \ldots, 3, 4, 2].
\]
Thus there is a crosscut in the cluster and
 \[
 A\left[|A| - i + 1\right] = A[i] + 1 \quad 1 \le i \le \frac{|A| - n + 1}{2},
 \]
 where $|A|$ is the length of $A$. Using similar reasoning to the first part of the proof,
   it can be shown that at each pair of crossing points corresponding to $A[i]$ and $A[|A|-i+1]$, the same pair  of curves intersect.
So, for each curve $C$ in the cluster as we move the rays along $C$ in opposite directions,
 we encounter the same curves that intersect $C$ and therefore the diagram is crosscut symmetric.
\qed
\end{proof}

Given a cluster $S_i$ of a crosscut symmetric $n$-Venn diagram where a segment of $C_i$ is the crosscut, there are the same number of regions on both sides of the crosscut. Furthermore, let $r$ be a $k$-region in $S_i$ that lies in the exterior of $C_i$ and interior to the curves in some
  set $\mathcal{K}$; then there is a corresponding $(k+1)$-set region $r'$ that is in the interior of $C_i$ and also interior
  to the curves in $\mathcal{K}$.

\section{Simple symmetric 11-Venn diagrams}
By Theorem \ref{Thm:Crosscut}, having the subsequence $\alpha$ of a crossing sequence we can construct the corresponding simple monotone crosscut symmetric $n$-Venn diagram. Therefore, for small values of prime $n$, an exhaustive search of $\alpha$ sequences may give us possible crosscut symmetric $n$-Venn diagrams. For example, for $n=3$ and $n=5$, $\alpha$ is empty and the only possible cases are the three circles Venn diagram and \Gr 5-ellipses. For $n=7$, the valid cases of $\alpha$ are $[3,2,4,3]$ and $[3,2,3,4]$.
The search algorithm is of the backtracking variety; for each possible case of $\alpha$, we construct the crossing sequence
  $S=\rho,\alpha, \delta, \alpha^{r+}$ checking along the way whether it currently satisfies the Venn diagram constraints, and
  then doing a final check of whether $S$ represents a valid rotationally symmetric Venn diagram.
Using this algorithm for $n=11$, we found more than $200,000$ simple monotone symmetric Venn diagrams which settles a long-standing open problem in this area.

Figure \ref{Fig:Newroz} shows the first simple symmetric 11-Venn diagram discovered.
It was discovered in March of 2012, so following Anthony Edwards' tradition of naming symmetric diagrams \cite{Ed_1998},
  we name it Newroz which means ``the new day" or ``the new sun" and refers to the
  first days of spring in Kurdish/Persian culture; for English speakers, Newroz sounds also like ``new rose", perhaps
  also an apt description.
Further illustrations of our 11-Venn diagrams may be found at
  \url{http://webhome.cs.uvic.ca/~ruskey/Publications/Venn11/Venn11.html}.
Below is the $\alpha$ sequence for Newroz.
\begin{align*}
[ &3 2 3 4 3 4 5 4 3 2 3 4 3 4 5 4 3 4 5 4 5 6 5 4 5 6 5 6 7 6 5 4 3 2 5 4 3 4 6 5 4 5 \\
  &6 7 6 7 8 7 6 5 6 5 4 3 4 5 7 6 5 4 6 5 8 7 6 5 4 5 7 6 5 6 8 7 6 5 4 6 5 7 6 5 6 7 ]
\end{align*}

\begin{figure}[t]
\centering{
\includegraphics[scale = 0.40]{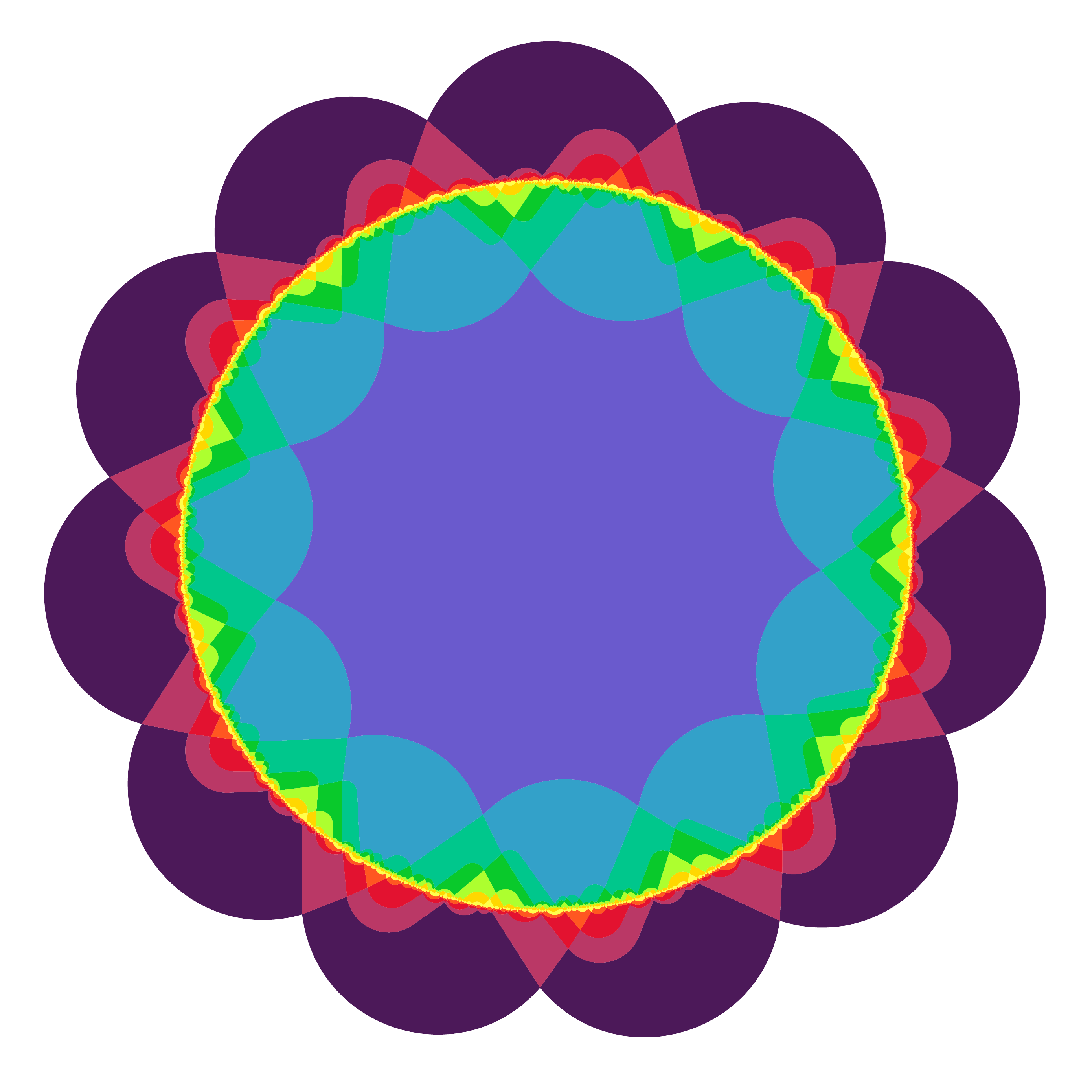}
\caption{Newroz, the first simple symmetric 11-Venn diagram.}
\label{Fig:Newroz}
}
\end{figure}

\begin{figure}
\centering{
\includegraphics[scale = 0.40]{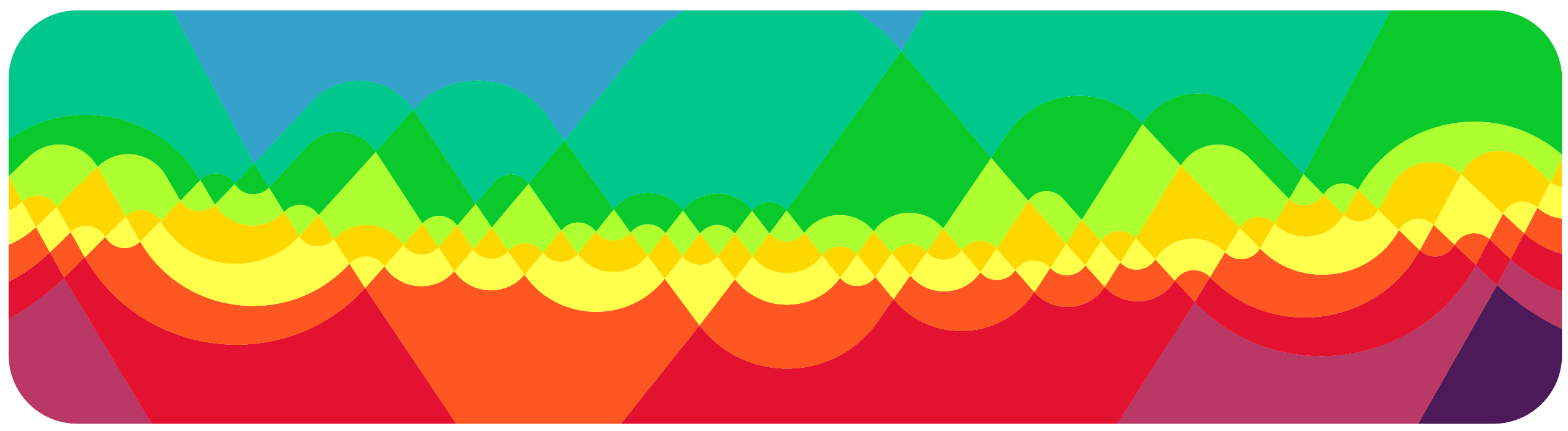}
\caption{A blow-up of part of Newroz.}
\label{Fig:Newroz_blowup}
}
\end{figure}

\begin{figure}
\centering{
\includegraphics[scale = 0.4]{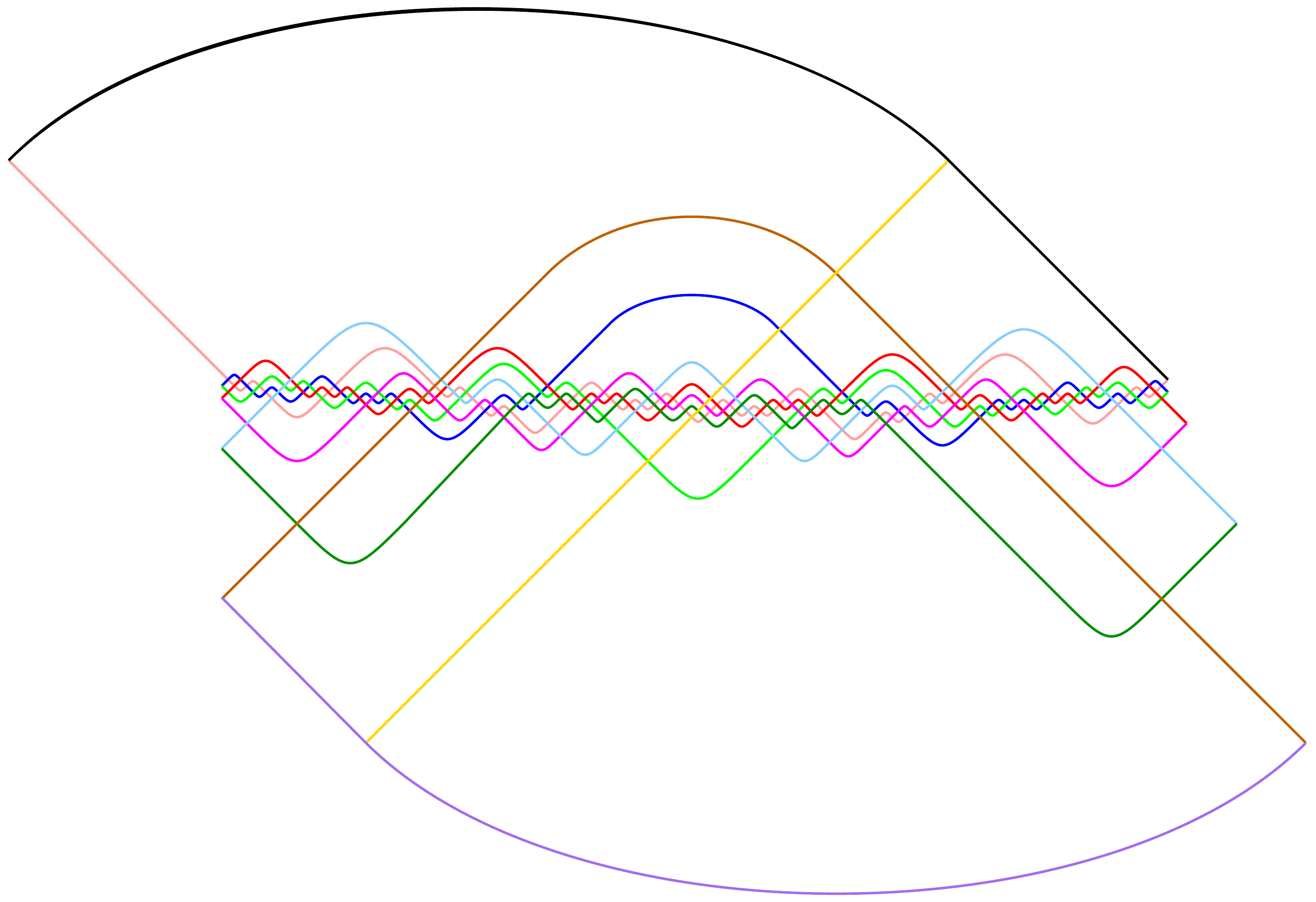}
}
\caption{A cluster of Newroz, the first simple symmetric 11-Venn diagram.}
\label{Fig:NewrozSlice}
\end{figure}

\section{Final Thoughts}

Define a symmetric Venn diagram to have \emph{dihedral symmetry} if it is
left invariant by flipping it over (and possibly rotating it after the flip).
It is interesting how close our diagram is to having dihedral symmetry.
By removing $n(n-1)/2 = 5$ edges, the remaining diagram can be drawn with dihedral
symmetry (Figure \ref{Fig:NewrozDihedral}(a)); the edges to be removed are the
diagonal edges that intersect a vertical bisector of the figure.
By slightly modifying those removed edges we can get a Venn
diagram with dihedral symmetry (Figure \ref{Fig:NewrozDihedral}(b)).
However, in the diagram of Figure \ref{Fig:NewrozDihedral}(b) there are curves that
intersect at infinitely many points, and there are exactly $n(n-3)/2 = 4$ points where 3 curves meet.

\begin{figure}
\includegraphics[scale = 0.3333]{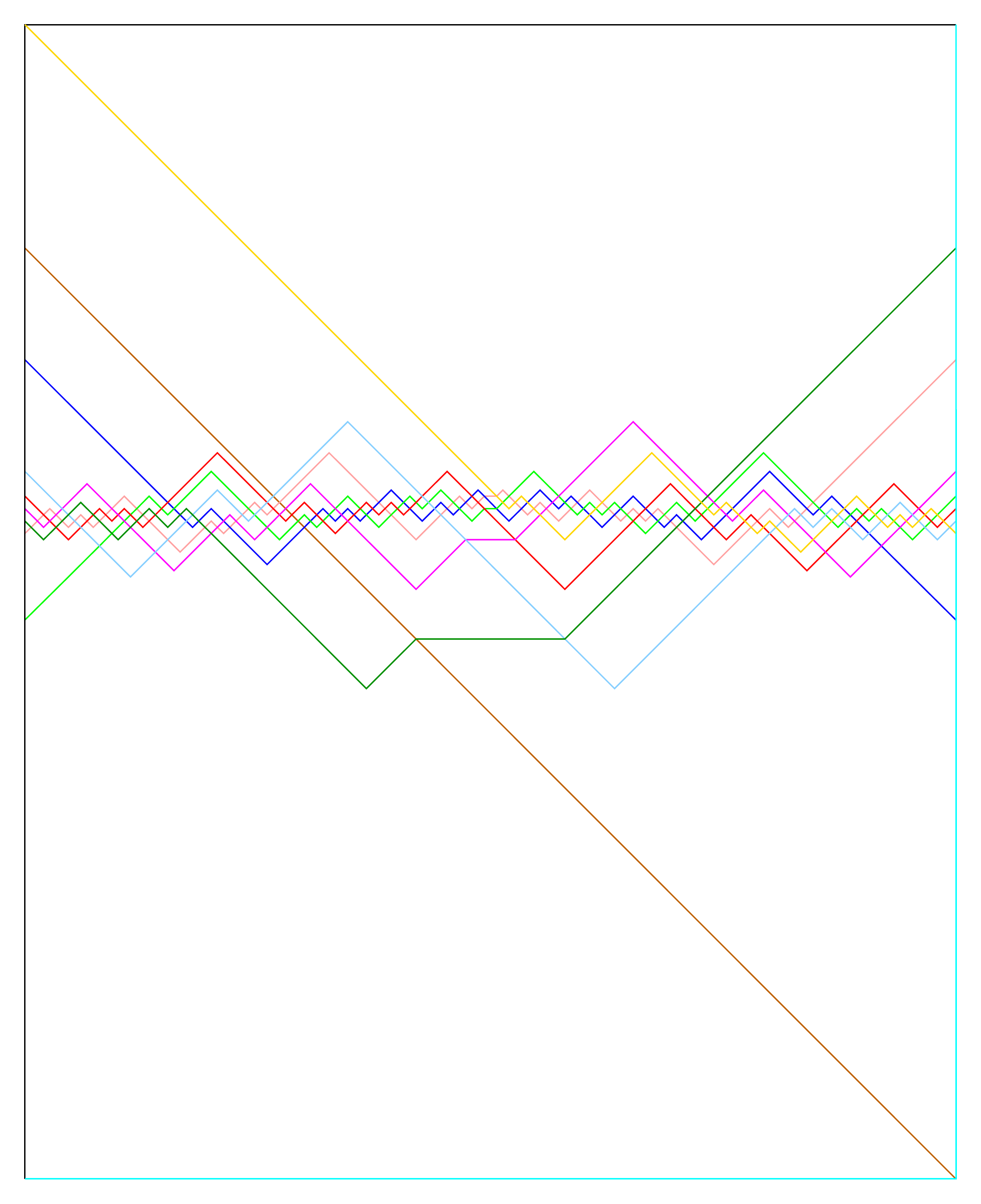}
\includegraphics[scale = 0.3333]{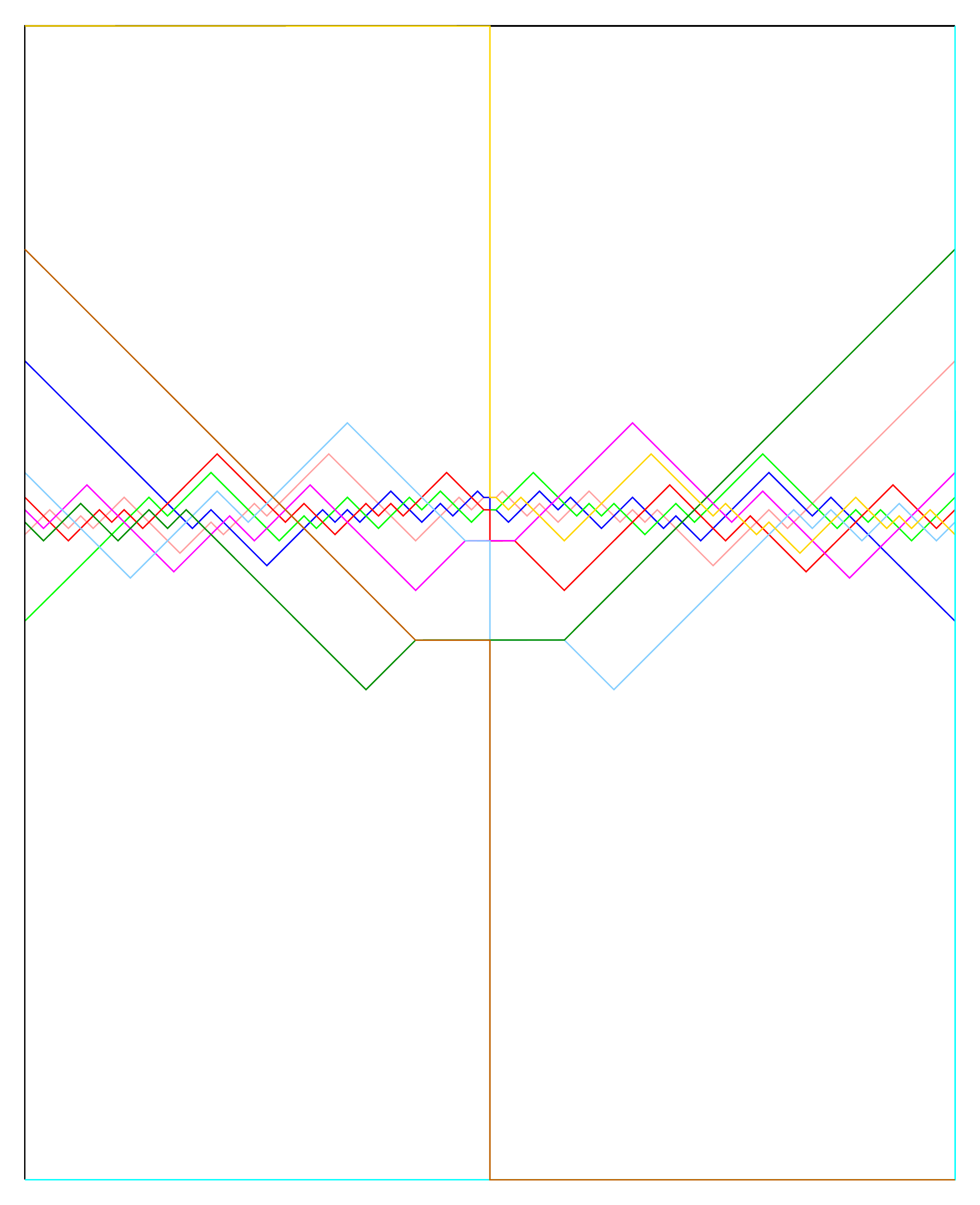}
\caption{(a) Part of two adjacent clusters of Newroz with the crosscut drawn vertically on
  the left and on the right.  The zig-zag starts at the diagonal in the NW corner, and exits
  at the diagonal in the SE corner.  (b) A slightly modified version, showing dihedral symmetry.}
\label{Fig:NewrozDihedral}
\end{figure}

\begin{conjecture}
There is no simple Venn diagram with dihedral symmetry if $n > 3$.
\end{conjecture}

\begin{figure}
\centering{
\includegraphics[scale=0.25]{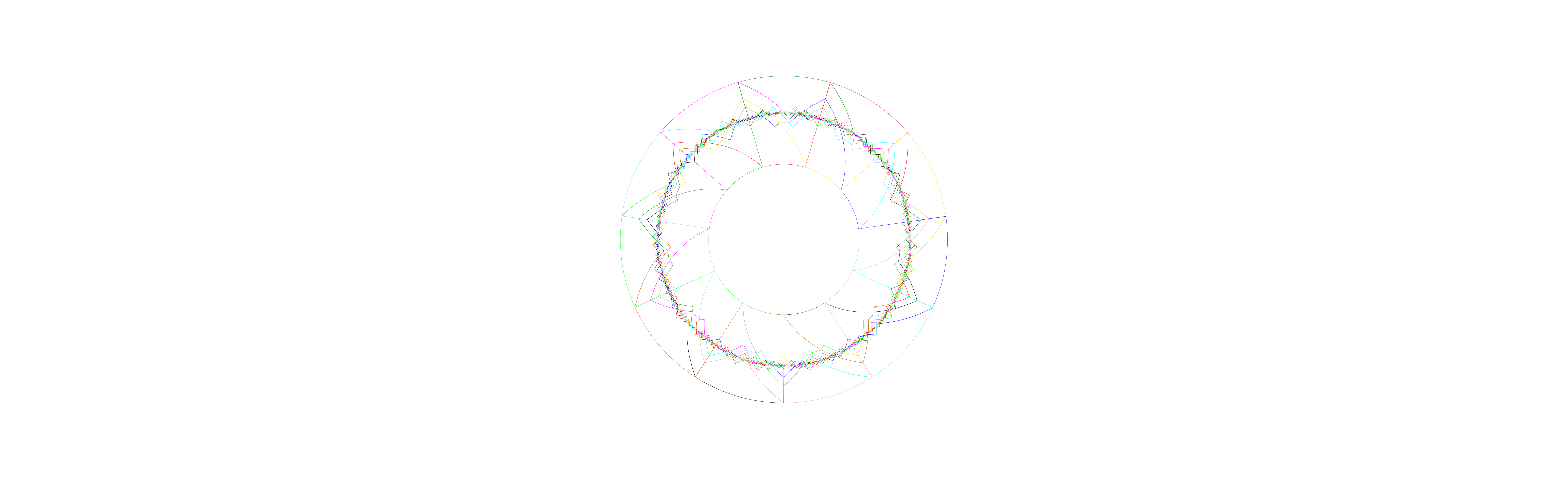}
}
\caption{Newroz drawn with the crosscuts drawn along rays emanating from the center of the diagram, and
with as much dihedral symmetry as possible.}
\label{Fig:NewrozDihedralFull}
\end{figure}

It is natural to wonder whether our method can be used to find a simple symmetric 13-Venn diagram.
So far our efforts have failed.  One natural approach is to try to exploit further symmetries of the
diagram.  In the past, attention has focused on what is known as polar symmetry --- in the cylindrical
representation, this means that the diagram is not only rotationally symmetric but that it is
also symmetric by reflection about a horizontal plane, followed by a rotation; an equivalent definition,
which we use below is that there is a axis of rotation through the equator that leaves the diagram fixed.
In the past polar symmetry did not help in finding 11-Venn diagrams (in fact, no polar symmetric Venn diagrams are yet known), but it is natural
to wonder whether it might be fruitful to search for diagrams that are both polar symmetric and crosscut
symmetric.  However, the theorem proven below proves that there are no such diagrams for primes larger than 7.

Of the four regions incident to an intersection point in a simple Venn diagram, a pair of non-adjacent regions are in the interior of the same number of curves and the number of containing curves of the other two differs by two.   Define a \emph{$k$-point} in a simple monotone Venn diagram to be
   an intersection point that is incident to two $k$-regions (and thus also one $(k-1)$-region and one $(k+1)$-region).
\begin{lemma}\label{Lem:K_Point}
Given a cluster of a simple symmetric monotone $n$-Venn diagram which also has crosscut symmetry, the number of $k$-points on the left side of the crosscut, for $1 \le k < n$ is
\[
\frac{1}{n}\left(\binom{n-1}{k}+(-1)^{k+1}\right).
\]
\end{lemma}
\begin{proof}
Let $R_k$ denote the number of $k$-regions on the left side of the crosscut. There is a one-to-one correspondence between the $k$-regions on the left side and $(k+1)$-region on the right side of the crosscut and we know that the total number of $k$-regions in the cluster is ${n \choose k} / n$. Therefore, we can compute the number of $k$-regions on the left side of the crosscut from the following recursion.
\begin{equation}\label{Eq:R}
R_k = \begin{cases}
		1,& \text{if } k=1\\
		\binom{n}{k}/n - R_{k-1}, & \text{if } 1 < k < n-1.
	  \end{cases}
\end{equation}
Unfolding the recurrence (\ref{Eq:R}), we have $nR_k = \sum_{0 \le j \le k-1} (-1)^j{n \choose k-j} = \binom{n-1}{k}+(-1)^{k+1}$.
Every $k$-point in a simple monotone Venn diagram indicates the ending of one $k$-region and starting of another one. Therefore, we have the same number of $k$-points and $k$-regions on the left side of the crosscut, so both of these are counted by $R_k$.
\qed
\end{proof}

\begin{theorem}
There is no monotone simple symmetric $n$-Venn diagram with crosscut and polar symmetry for $n > 7$.
\end{theorem}
\begin{proof}
Let $V$ be a monotone simple symmetric $n$-Venn diagram which has been drawn in the cylindrical representation with both polar and crosscut symmetry,
and let $S$ be a cluster of $V$ with crosscut $C$.  Since the diagram is polar symmetric, $S$ remains fixed under a rotation of $\pi$ radians about some axis through the equator.  Under the polar symmetry action, a crosscut must map to a crosscut, and since there are $n$ crosscuts and $n$ is
odd, one crosscut must map to itself, and so one endpoint of this axis can be taken to be the central point, call it $x$, of some crosscut (the other endpoint of the axis will then be midway between two crosscuts).

Now, consider a horizontal line $\ell$ that is the equator (and thus must include $x$).
Let $m = (n-1)/2$.  The line $\ell$ cuts every $m$-region on the left side of $C$ and every $(m+1)$-region on the right side of the crosscut. Because of the crosscut symmetry of $V$, each $(m+1)$-region on the equator (and so in the interior of $C$) corresponds to the image of some $m$-region on the
equator (and so in the exterior of $C$) with the same number of bounding edges. Therefore, for $V$ to be polar-symmetric, every $m$-region on the left side of the crosscut must be symmetric under a flip about the horizontal line $\ell$, and thus each such region must have an even number of bounding edges.  The number of bounding edges can not be 2, because otherwise it is not a Venn diagram.
Thus a $m$-region on $\ell$ must contain at least one $(m-1)$-point.  One of those $(m-1)$-points might lie on $C$ and thus not be counted by
   $R_{m-1}$, but all other points are counted and so we must have $R_m \le R_{m-1}+1$.  However,
\begin{align*}
R_m - R_{m-1}
&= \frac{1}{n} \left( {2m-2 \choose m} - {2m-2 \choose m-1} + 2 (-1)^{m+1} \right) \\
&= \frac{1}{n} ( c_{m-1} + 2(-1)^{m+1} ),
\end{align*}
where $c_m$ is a Catalan number.  An easy calculation then shows that $R_m \le R_{m-1}+1$ only for the primes $n \in \{2,3,5,7\}$.
\qed
\end{proof}

\section*{Acknowledgements}

The authors are grateful to  Mark Weston at the University of Victoria and to
  Rick Mabry at Louisiana State University in Shreveport for independently verifying that the
  symmetric 11-Venn diagram illustrated in this paper is correct.
We also thank Branko Gr\"{u}nbaum and Anthony Edwards for supplying us with historical background.

\bibliographystyle{plain}

\bibliography{VennBib}{}

\end{document}